\documentclass[letterpaper, 10 pt, conference]{ieeeconf}
\IEEEoverridecommandlockouts
\usepackage{amsmath,bm}
\usepackage{amsfonts}
\usepackage{amssymb}
\usepackage{graphicx}
\usepackage[dvipsnames]{xcolor}
\makeatletter
\let\MYcaption\@makecaption
\makeatother

\usepackage[font=footnotesize]{subcaption}

\makeatletter
\let\@makecaption\MYcaption
\makeatother

\usepackage{hyperref}

\newtheorem{Theorem}{Theorem}

\newtheorem{Lemma}{Lemma}
\newtheorem{Problem}{Problem}
\newtheorem{Corollary}{Corollary}
\newtheorem{Remark}{Remark}
\newtheorem{Assumption}{Assumption}
\newcommand{\m}[1]{\mathbf{#1}}
\newcommand{\mc}[1]{\mathcal{#1}}
\newcommand{\mb}[1]{\mathbb{#1}}

\begin{document}

\title{\LARGE \bf Consensus seeking in diffusive multidimensional networks with a repeated interaction pattern and time-delays}
\author{Hoang Huy Vu\authorrefmark{1}, Quyen Ngoc Nguyen\authorrefmark{1},  Tuynh Van Pham\authorrefmark{1}, Chuong Van Nguyen\authorrefmark{2}, Minh Hoang Trinh\authorrefmark{3}
\thanks{${}^{*}$Department of Automation Engineering, School of Electrical and Electronic Engineering, Hanoi University of Science and Technology (HUST), Hanoi, Vietnam. E-mails: \tt\footnotesize{\{hoang.vh242063m,quyen.nn200518\} @sis.hust.edu.vn, 
tuynh.phamvan@hust.edu.vn}.}
\thanks{${}^{\dagger}$Viterbi School of Engineering – Department of Aerospace and Mechanical Engineering, University of Southern California, USA. E-mail: \tt\footnotesize{vanchuong.nguyen@usc.edu}.}
\thanks{${}^{\ddagger}$AI Department, FPT University, Quy Nhon AI Campus, An Phu Thinh New Urban Area, Quy Nhon City, Nhon Binh Ward, Binh Dinh 55117, Vietnam. Corresponding author. E-mail: \tt\footnotesize{minhth30@fpt.edu.vn}.}
}

\maketitle
\thispagestyle{empty}
\begin{abstract}
This paper studies a consensus problem in multidimensional networks having the same agent-to-agent interaction pattern under both intra- and cross-layer time delays. Several conditions for the agents to asymptotically reach a consensus are derived, which involve the overall network's structure, the local interacting pattern, and the assumptions specified on the time delays. The validity of these conditions is proved by direct eigenvalue evaluation and supported by numerical simulations.
\end{abstract}

\section{Introduction}
\label{sec:intro}
In recent years, the demand for understanding complex and large-scale structures such as social, traffic, material, or brain networks has raised increasing attention on modeling, analyzing, and synthesizing multilayer networks. A multilayer network consists of multiple agents (or subsystems) interacting via $d$-dimensional single-layer networks $(d\ge 2)$. The dynamic of each agent is captured by a set of state variables corresponding to $d$ layers, and the agent-to-agent influences govern the entire network's behavior. If the network contains only \emph{intra-layer interactions} (interactions of state variables belonging to the same layer) and \emph{inter-layer interactions} (couplings of state variables from different layers of the same agent), we refer to the network as a \emph{multiplex}. A multilayer network contains \emph{cross-layer interactions} - the couplings of the state variables associated with different network layers and from distinct agents. In a diffusive network, the variation of each agent's state variables depends on the differences between the agent's state and its neighboring agents. Figure~\ref{fig:2layerNetw} illustrates a two-layer network of four agents.

Consensus algorithms, because of their simplicity and generality, were extensively used to study the dynamics of single-layer networks (monoplexes) \cite{Olfati2007consensuspieee}. For multilayer networks, \cite{He2017TSMC,Lee2017consensus,Sorrentino2020group} considered the consensus and synchronization on multiplexes. The authors in  \cite{Trinh2018Aut} proposed matrix-weighted consensus in which the state variable in a layer of an agent is updated based on a weighted sum of relative states from every layer, taken for all neighboring agents. Matrix-weighted consensus has found applications in modeling of multidimensional opinion dynamics and networked economic network, see e.g. \cite{Ahn2020opinion,Ye2020Aut,Trinh2024networked}. Time delay is a source of uncertainty that usually negatively affects the performance of communication networks \cite{Fridman2014tutorial}. Consequently, considerable attention has been devoted to the analysis of delayed single-layer consensus networks~
\cite{Olfati2007consensuspieee,Sun2008average,Cepeda2011exact}. In contrast, the investigation of time-delay effects in multilayer consensus systems remains relatively limited. Related works include the study of consensus and synchronization in time-delayed multiplex networks presented in \cite{Singh2015synchronization}, and the analysis of matrix-weighted consensus algorithms under time delays in \cite{Lam2024consensus}.

\begin{figure}[t]
    \centering
   \includegraphics[width=0.8\linewidth]{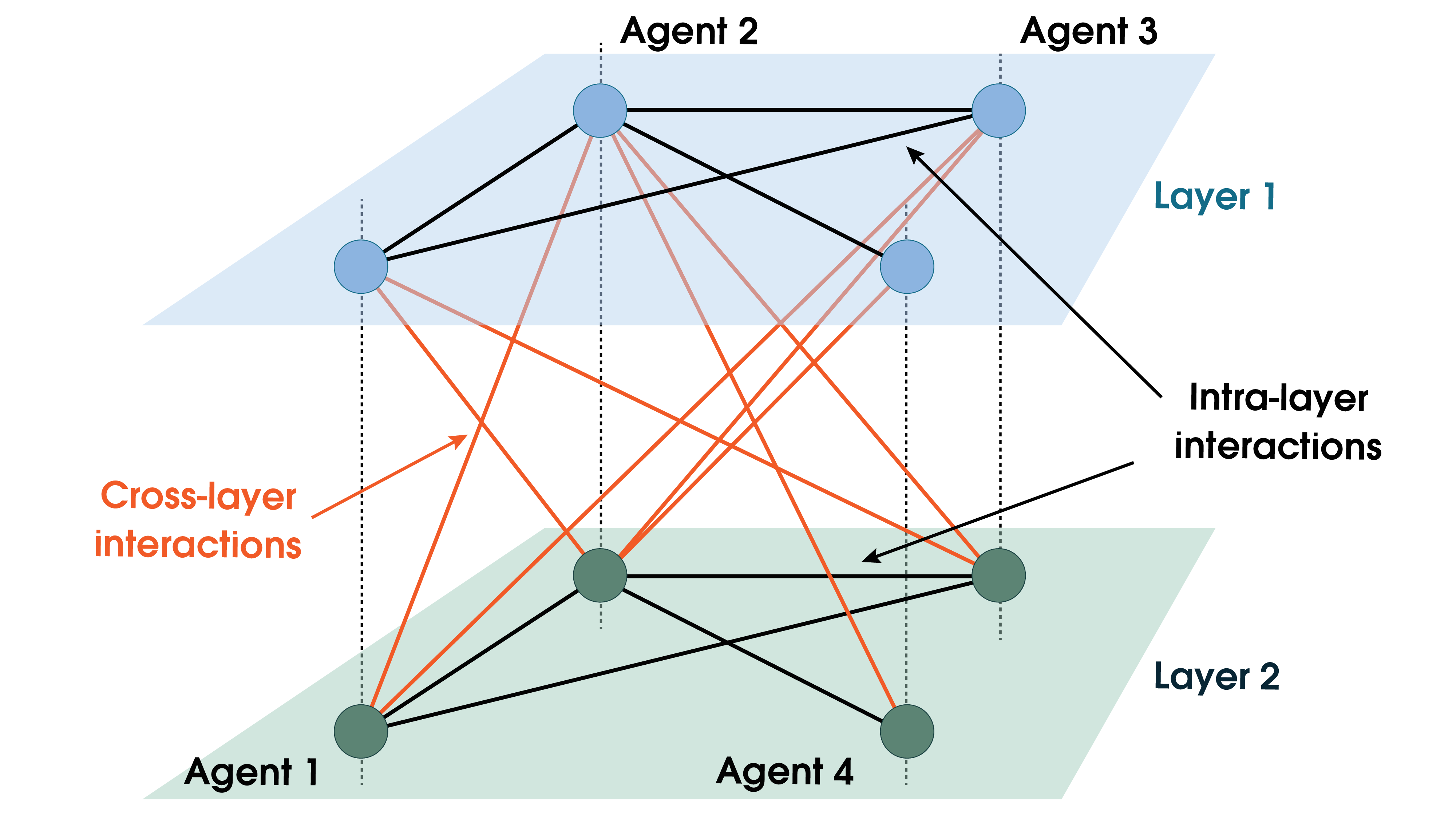}
   \caption{A two-layer network of four agents}
   \label{fig:2layerNetw}
\end{figure}
This paper attempts to derive consensus conditions for the matrix-weighted consensus model \cite{Trinh2018Aut} with heterogeneous time delays. We assume that all agent-to-agent interactions have the same graph pattern and time delays may exist in both intra- and cross-layer interactions. We first prove that if there is no intra-layer time delay and the maximum magnitude of all eigenvalues of the adjacency matrix corresponding to cross-layer interactions is smaller than unity, the network globally asymptotically achieves a consensus. The result reveals that significant delays from weak cross-layer interactions cannot destabilize a multilayer consensus network. Second, the network is considered with both intra- and cross-layer time delays. The stability of an equation involving two different constant time delays was considered in \cite{Ruan2003OnTZ}, where the authors derived a necessary and sufficient condition based on the root of a transcendental equation. In this paper, due to the interweaving of the overall network and the local interaction graph, the method developed in \cite{Ruan2003OnTZ} cannot be applied. Instead, we examine the network when two-time delays are equal and determine a delay margin. Then, we show that when intra- and cross-layer time delays do not exceed a calibrated delay margin, the network asymptotically achieves a consensus. The delay margin is tight in the sense that once one of the time delays exceeds the margin and the other is equal to the delay margin, the consensus network is unstable. Finally, the delayed two-layer network is considered and several detailed consensus conditions are determined by the graph's parameters.

The remainder of this paper is organized as follows. Section~\ref{sec:modeling} contains theoretical background and problem formulation. Section~\ref{sec:consensus_condition} provides consensus conditions and the corresponding analysis for networks with a general interaction pattern matrix. Then, two-layer networks are considered and simulated in Section~\ref{sec:twolayer_network}. Lastly, Section~\ref{sec:conclusion} concludes the paper.

\section{Problem formulation}
\label{sec:modeling}
Consider a diffusive multi-layer network of $n\ge2$ agents. Each agent $i\in \{1,\ldots,n\}$ is a $d$-dimensional subsystem with the state vector $\m{x}_i=[x_{i1},\ldots,x_{id}]^\top \in\mb{R}^d$ $(d\geq 2)$. The interaction between two agents $i$ and $j$ in the network is modeled by a matrix weight $\m{A}_{ij}=\m{A}_{ji} \in \mb{R}^{d\times d}$, 
{\small
\begin{align}
    \m{A}_{ij} = [a_{ij}^{pq}]_{d\times d} =
    \begin{bmatrix}
        a_{ij}^{11} & a_{ij}^{12} & \ldots & a_{ij}^{1d}\\
        a_{ij}^{21} & a_{ij}^{22} & \ldots & a_{ij}^{2d}\\
        \vdots & \vdots &  & \vdots \\
        a_{ij}^{d1} & a_{ij}^{d2} & \ldots & a_{ij}^{dd}
    \end{bmatrix},
\end{align}}
where $i, j \in \{1,\ldots,n\}$, $i\neq j$, and $p, q \in \{1,\ldots,d\}$. Thus, each element $a_{ij}^{pq} = a_{ji}^{pq} \in \mb{R}$ captures the influence from the $p$-th layer to the $q$-th layer in the network. The elements $a_{ij}^{pp},~p=1,\ldots,d,$ represent intra-layer (same layer) interactions, while the elements $a_{ij}^{pq},~p\neq q$ stand for cross-layer interactions between two agents $i,~j$. We use an undirected graph $\mc{G}=(\mc{V},\mc{E})$ to describe the interaction topology between agents in the network, i.e., each agent $i$ is represented by a vertex $i$ in the vertex set $\mc{V}=\{1,\ldots,n\}$, each edge $(i,j) \in \mc{E}$ exists if the matrix weight $\m{A}_{ij} \neq \m{0}_{d\times d}$. Assume that the graph $\mc{G}$ does not exist any self-loop, i.e., edges connecting a vertex with itself. The undirectedness assumption implies that if $(i,j) \in \mc{E}$, then $(j,i)\in \mc{E}$. Let $\mc{N}_i=\{j \in \mc{V}|(i,j)\in \mc{E}\}$ denote the neighbor set of vertex $i$. A path $i_1i_2\ldots i_k$ in $\mc{G}$ is a sequence of vertices $i_r \in \mc{V},~\forall r=1,\ldots,k,$ joining the starting vertex $i_1$ to the end vertices $i_k$ by $k-1$ edges $(i_r,i_{r+1}) \in \mc{E}$, $r=1,\ldots,k-1$. The graph $\mc{G}$ is \emph{connected} if and only if for any pairs of vertices in $\mc{V}$, there exists a path joining them.

Define the network adjacency matrix $\m{W} = [w_{ij}] \in \mb{R}^{n \times n}$ of $\mc{G}$ with elements $w_{ij}=1$ if $(i,j) \in \mc{E}$ and $w_{ij}=0$, otherwise. Then, the \emph{Laplacian matrix} of $\mc{G}$ is defined as $\m{L} = [l_{ij}] = \text{diag}(\m{W}\m{1}_n)-\m{W} \in \mb{R}^{n\times n}$.

\begin{figure}
    \centering
    \subfloat[$\mc{G}$]{\includegraphics[height=2.5cm]{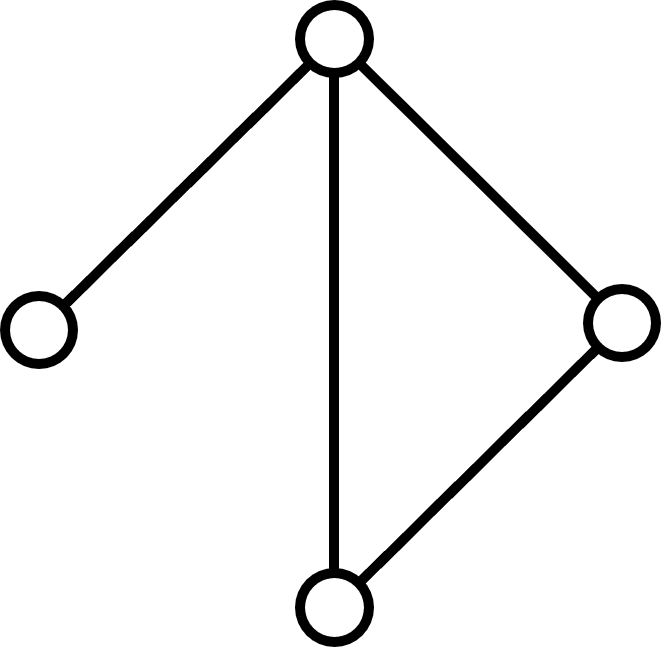}}\qquad\qquad
    \subfloat[$\mc{H}$]{\includegraphics[height=2.5cm]{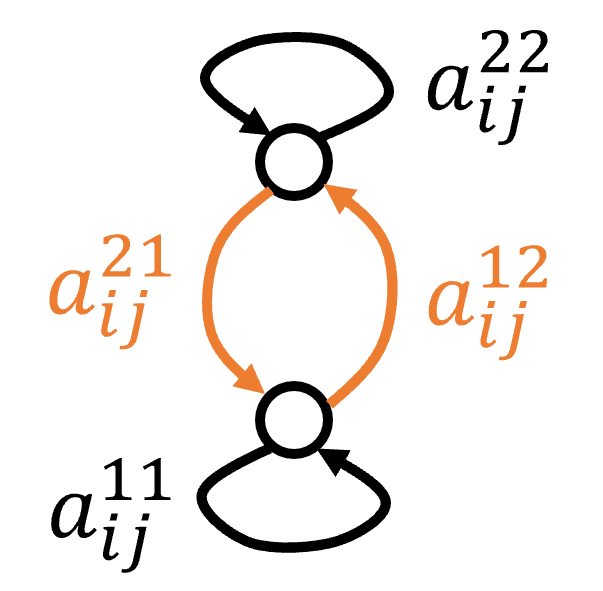}}
    \hfill
    \caption{(a) - The network graph $\mc{G}=(\mc{V},\mc{E})$, and (b) - the interaction pattern $\mc{H}$ corresponding to the two-layer network in Fig.~\ref{fig:2layerNetw}.}
    \label{fig:G-and_H}
\end{figure}

We consider the consensus algorithm in the multi-layer network, where each agent updates its state variables according to a weighted sum of both intra-layer and cross-layer state differences. Moreover, it is supposed that there are time-delays in both the intra- and cross-layer interactions specified by $\tau_1, \tau_2 \geq 0$, respectively. Subsequently, the equation governing the network's evolution under the matrix-weighted consensus algorithm is given as
\begin{align} \label{eq:consensus_with_delay}
    \dot{x}_{iq}(t) &= \underbrace{\sum_{j=1}^n a_{ij}^{qq}(x_{jq}(t-\tau_1) - x_{iq}(t-\tau_1))}_{\text{intra-layer interactions}} \nonumber\\ 
    &\quad+ \underbrace{\sum_{j=1}^n\sum_{p=1,p\neq q}^d a_{ij}^{pq}(x_{jp}(t-\tau_2) - x_{ip}(t-\tau_2))}_{\text{cross-layer interactions}},
\end{align}
for all $i\in \mc{V}$ and $q \in \{1,\ldots, d\}$. 

We assume that the multilayer network can be represented by a pair of graphs: the network graph $\mc{G}$ (the global graph), and an agent-to-agent interaction graph $\mc{H}$ (the local graph). The graph $\mc{H}$ encodes the mutual influences between the relative states of any two adjacent agents in the network. The multi-layer networks with a repeated pattern are described in the following assumptions.
\begin{Assumption} \label{assumption:1}
    The network graph $\mc{G}$ is undirected and connected.
\end{Assumption}

\begin{Assumption} \label{assumption:2}
    The interaction matrices are the same $\m{A}_{ij}=\m{A}=[a^{pq}]_{d\times d}$ for all edges $(i,j)\in \mc{E}$. Furthermore, all intra-layer interactions have the same weight $a_{ij}^{qq}=1,~\forall (i,j)\in \mc{E},~\forall q=1,\ldots,d$.
\end{Assumption}

Under Assumptions~\ref{assumption:1}--\ref{assumption:2}, we have $\m{A}_{ij}=\m{A}_{ji}=\m{A}=[a_{ij}]_{d\times d}$ and $\mc{H}$ is a graph of $d$ vertices with self-loops. Moreover, as all agent-to-agent interactions are represented by the same graph $\mc{H}$, we will refer to $\mc{H}$ as the \emph{interaction pattern} of the network. Figure~\ref{fig:2layerNetw} illustrates a four-agent network consisting of the network graph $\mc{G}$ and the interaction pattern $\mc{H}$.

Note that the consensus algorithm \eqref{eq:consensus_with_delay} without time delays has been studied in \cite{Trinh2018Aut} under the assumption of symmetric cross-layer interactions and positive semi-definiteness of the interaction matrix ($\m{A}_{ij}=\m{A}_{ij}^\top\geq 0$). In \cite{Ahn2020opinion}, the delay-free model \eqref{eq:consensus_with_delay} is considered when the matrix weights $a_{ij}^{pq}=a_{ji}^{pq}$ is a function of the relative state variable $(x_{jp}-x_{iq})$ and several convergence results were derived.
 
The problem studied in this paper can be stated as follows.
\begin{Problem}
    Consider a network satisfying the Assumptions~\ref{assumption:1} and \ref{assumption:2}. Determine the condition for the system \eqref{eq:consensus_with_delay} in terms of the intra- and cross-layer time delays $\tau_1$ and $\tau_2$ to asymptotically achieve a consensus.
\end{Problem}

\section{Consensus conditions}
\label{sec:consensus_condition}
In this section, we consider the delayed consensus network~\eqref{eq:consensus_with_delay} and derive consensus conditions under different assumptions of the time-delays.
\subsection{Delay-free network}
To derive a consensus condition, the system without time delays $\tau_1 = \tau_2 = 0$ is investigated. The following theorem gives a necessary and sufficient consensus condition. 
\begin{Theorem}
    Suppose that Assumptions~\ref{assumption:1} and \ref{assumption:2} hold. The $n$-agent system \eqref{eq:consensus_with_delay} with $\tau_1=\tau_2 = 0$ achieves a consensus if and only if the interaction matrix $-\m{A}$ is Hurwitz.
\end{Theorem}
\begin{proof}
    The $n$-agent system can be represented in matrix form as follows
\begin{align}
    \dot{\m{x}}(t) = -(\m{L}\otimes \m{A})\m{x}(t),
\end{align}
where $\m{x}=\text{vec}(\m{x}_1,\ldots,\m{x}_n)$. Let $\m{U}$ be the orthonormal matrix that diagonalizes $\m{L}$, i.e., $\m{L}= \m{U}\m{\Lambda}\m{U}^\top$, where $\text{diag}(0,\lambda_2,\ldots,\lambda_n)$, and $0<\lambda_2\leq \ldots \leq \lambda_n$ are eigenvalues of $\m{L}$. Let $\zeta_1,\ldots,\zeta_d$ be eigenvalues of $\m{A}$, and suppose that $\m{A}$ has a Jordan normal form $\m{A}=\m{T}_{\rm A}\m{J}_{\rm A}\m{T}^{-1}_{\rm A}$. Then,  matrix $\m{L}\otimes \m{A} = (\m{U}\bm{\Lambda}\m{U}^\top) \otimes (\m{T}_{\rm A}\m{J}_{\rm A}\m{T}^{-1}_{\rm A}) = (\m{U}\otimes\m{T}_{\rm A})(\bm{\Lambda\otimes\m{J}_{\rm A}})(\m{U}^\top \otimes \m{T}^{-1}_{\rm A})$ has eigenvalues $\lambda_i \eta_j$, $i=1,\ldots,n,$ and $j=1,\ldots,d$. Taking the change of variables $\m{z}=(\m{U}^\top \otimes \m{T}^{-1}_{\rm A})\m{x}$, it follows that
\begin{align} \label{eq:z_variable}
    \dot{\m{z}}(t)= -(\bm{\Lambda}\otimes \m{J}_{\rm A})\m{z}(t),
\end{align}
Noting that $[z_1,\ldots,z_d]^\top$ corresponds to the consensus space im$(\m{1}_n\otimes\m{I}_d)$, and these variables remain unchanged under \eqref{eq:z_variable}. Meanwhile, $\m{z}'=[z_{d+1},\ldots,z_{dn}]^\top$ corresponds to the disagreement space, $\dot{\m{z}}'=-(\bm{\Lambda}'\otimes \m{A}){\m{z}}'$, and $\bm{\Lambda}'={\rm diag}(\lambda_2,\ldots,\lambda_n)$. 
Thus, if at least an eigenvalue $\zeta_k$ of $\m{A}$ has a nonpositive real part, then $-\lambda_i \zeta_k$ has nonnegative real part. It follows that $\m{z}'$ may not converge to zero or even grow unbounded. In contrast, if all eigenvalues of $-\m{A}$ has negative real parts, then $-(\bm{\Lambda}'\otimes \m{A})$ is Hurwitz, $\m{z}'(t)$ exponentially converges to $\m{0}_{dn-d}$. Equivalently, the system exponentially achieves a consensus if and only $-\m{A}$ is Hurwitz.
\end{proof}

\subsection{Network with cross- and intra-layer time-delays}
We consider the matrix-weighted consensus network in the presence of under delays. Using the notation $\m{A}_{\rm cross} = \m{A} - \text{diag}(a^{11},\ldots,a^{dd}),$ the equation \eqref{eq:consensus_with_delay} can be written for each agent $i$ as follows
\begin{align}
    \dot{\m{x}}_i(t) &= \sum_{j=1}^n (\m{x}_j(t-\tau_1)-\m{x}_i(t-\tau_1)) \nonumber \\& \quad+ \sum_{j=1}^n \m{A}_{\rm cross} (\m{x}_j(t-\tau_2) -\m{x}_i(t-\tau_2)),
\end{align}
for $i=1,\ldots,n$. 

The $n$-agent network can be  written as follows
\begin{align} \label{eq:netw_matrix_form}
    \dot{\m{x}}(t) = - (\m{L}\otimes \m{I}_d)\m{x}(t-\tau_1) - (\m{L}\otimes \m{A}_{\rm cross})\m{x}(t-\tau_2),
\end{align}
Taking the Laplace transformation of the system~\eqref{eq:consensus_with_delay} under the assumption that all initial conditions are zero gives
\begin{align*}
    (s\m{I}_{dn}+(\m{L}\otimes\m{I}_d) e^{-\tau_1 s}+ \m{L}\otimes\m{A}_{\rm cross}e^{-\tau_2 s})\m{X}(s)=\m{0}_{dn}.
\end{align*}
Since each zero of the matrix $s\m{I}_{dn}+\m{L}\otimes\m{I}_d e^{-\tau_1 s}+ (\m{L}\otimes\m{A}_{\rm cross})e^{-\tau_2 s}$ is correspondingly a pole of the system \eqref{eq:netw_matrix_form}, we have the following result on the consensus condition of the delay consensus system~\eqref{eq:netw_matrix_form}. 

\begin{Lemma}\label{lm:poly}
    The system~\eqref{eq:netw_matrix_form} globally asymptotically achieves a consensus if and only if the  polynomials $\text{det}((s + \lambda_i e^{-\tau_1 s})\m{I}_d + \lambda_i \m{A}_{\rm cross}e^{-\tau_2 s})$, $i=2,\ldots,n$ are Hurwitz.
\end{Lemma}
\begin{proof}
    The poles of \eqref{eq:netw_matrix_form} are roots of
\begin{align}
    \text{det}(s\m{I}_{dn}+(\m{L}\otimes\m{I}_d) e^{-\tau_1 s} + \m{L}\otimes\m{A}_{\rm cross}e^{-\tau_2 s}) &= 0 \nonumber\\
    \text{det}(s\m{I}_{dn}+(\m{\Lambda}\otimes\m{I}_d)e^{-\tau_1 s} + \m{\Lambda}\otimes\m{A}_{\rm cross}e^{-\tau_2 s}) &= 0 \nonumber\\
    s^d \prod_{i=2}^{n} \text{det}((s + \lambda_i e^{-\tau_1 s} )\m{I}_d + \lambda_i \m{A}_{\rm cross}e^{-\tau_2 s}) &= 0.\nonumber
\end{align}
Note that for the eigenvalue $s=0$, we can find $d$ independent eigenvectors, which are columns of $\m{1}_n\otimes \m{I}_d$. Since these eigenvectors span the consensus space, the eigenspaces corresponding to the remaining eigenvalues span the disagreement space. Thus, the system \eqref{eq:netw_matrix_form} globally asymptotically achieves a consensus if and only if each equation
\begin{align} \label{eq:characteristic_equation}
    \text{det}((s + \lambda_ie^{-\tau_1 s})\m{I}_d + \lambda_i \m{A}_{\rm cross}e^{-\tau_2 s}) = 0,
\end{align}
$i=2,\ldots,n$, admits only roots with negative real parts.
\end{proof}

Let $\mu_k = a_k + \jmath b_k$, $a_k, b_k \in \mb{R}$, $k=1,\ldots,d,$ be eigenvalues of $\m{A}_{\rm cross}$, and consider the Jordan decomposition of $\m{A}_{\rm cross}$ as $\m{A}_{\rm cross}=\m{T}\m{J}\m{T}^{-1}$, where $\m{T} \in \mb{R}^{d \times d}$. 

Lemmas \ref{lem:Hurwitz_A} and \ref{lem:3_2} will be used to derive a stability result for the system~\eqref{eq:netw_matrix_form}.
\begin{Lemma} \label{lem:Hurwitz_A}
    Suppose that Assumptions~\ref{assumption:1} and \ref{assumption:2} hold. If all eigenvalues $\mu_k$ of $\m{A}_{\rm cross}$ satisfy $|\mu_k| < 1,~\forall i=1,\ldots,d$, then $-\m{A}$ is Hurwitz.
\end{Lemma}
\begin{proof}
    Since $-\m{A}= -\m{I}_d - \m{A}_{\rm cross}$, the eigenvalues of $-\m{A}$ are correspondingly $-1-\mu_k$, which have nagative real parts due to $|\mu_k| < 1,~\forall i=1,\ldots,d$.
\end{proof}
\begin{Remark} Using the Gershgorin circle theorem, a sufficient condition for $|\mu_k|<1, \forall k$ is $\begin{Vmatrix}\m{A}_{\rm cross}\end{Vmatrix}_\infty<1$.
\end{Remark}
\begin{Lemma}\label{lem:3_2}\cite[Cor. 2.4]{Ruan2003OnTZ}
    Consider the quasi-polynomial 
\begin{align}
&P\left(\lambda, e^{-\lambda \tau_1}, \cdots, e^{-\lambda \tau_m}\right) \nonumber \\
= & \displaystyle\sum_{j=1}^n p_j^{(0)}\lambda^{n-j}
+ \displaystyle\sum_{i=1}^m\left( e^{-\lambda\tau_i}\displaystyle\sum_{j=1}^n p_j^{(i)}\lambda^{n-j}\right),
\end{align}
where $i=1,\ldots,m, j=1,\ldots,n$, $\tau_i \geq 0$ and $p_j^{(i)}$ are constants. As $\left(\tau_1, \tau_2, \cdots, \tau_m\right)$ varies, the sum of the orders of the zeros of $P\left(\lambda, e^{-\lambda \tau_1}, \cdots, e^{-\lambda \tau_m}\right)$ in the open RHP can change only if a zero appears on or crosses the imaginary axis.
\end{Lemma}

First, we consider the situation when intra-layer interactions are delay-free, i.e., $\tau_1 = 0$, and prove the following theorem. 
\begin{Theorem} \label{thm:no-interlayer-delay}
    Suppose that Assumptions~\ref{assumption:1} and \ref{assumption:2} hold, $\tau_1 = 0$, and $|\mu_k|<1$, $\forall k = 1,\ldots, d$. Then, the system~\eqref{eq:netw_matrix_form} globally asymptotically achieves a consensus.
\end{Theorem}
\begin{proof}
Each equation $\text{det}((s + \lambda_i)\m{I}_d + \lambda_i \m{A}_{\rm cross}e^{-\tau_2 s})=0,~i=1,\ldots,n,$ is equivalent to $d$ equations
\begin{align} \label{eq:stability_eq_general}
    f_{i,k}(s,e^{-\tau_2 s}) = s+ \lambda_i + \lambda_i (a_k + \jmath b_k) e^{-\tau_2 s}=0,
\end{align}
Notice that for the quasi-polynomial \eqref{eq:stability_eq_general}, based on Lemma~\ref{lem:Hurwitz_A}, the polynomial $f_{i,k}\left(s, e^{-\tau_2 s}\right)$ is Hurwitz stable for $\tau_2 = 0$ and $|\mu_k|^2=a^2_k+b^2_k<1$. From Lemma~\ref{lem:3_2}, suppose that $f_{i,k}\left(s, e^{-\tau_2 s}\right)$ is unstable, there must exist some $0<\tau^*<\tau_2$ such that $f\left(s, e^{-\tau^*s}\right)$ has a root on the imaginary axis. 

Let $s=\jmath \omega,~\omega\in \mb{R}$, be a root of Eqn.~\eqref{eq:stability_eq_general}, then
\begin{align}
    \jmath \omega + \lambda_i + \lambda_i (a_k+b_k\jmath)(\cos(\tau_2\omega)-\jmath\sin(\tau_2\omega)) = 0,
\end{align}
which is equivalent to
\begin{subequations}
    \begin{align}
    a_k\cos(\tau_2\omega) + b_k\sin(\tau_2\omega) &= -1, \label{eq:cos_sina1}\\
   -a_k\sin(\tau_2\omega) + b_k\cos(\tau_2\omega) &= \frac{\omega}{\lambda_i}. \label{eq:cos_sinb1}
    \end{align}
\end{subequations}
Taking the sum of square of both sides of two equations~\eqref{eq:cos_sina1} and \eqref{eq:cos_sinb1} gives $a^2_k+b^2_k = \frac{\omega^2}{\lambda_i^2} + 1$, which has no real roots as $a^2_k+b^2_k-1<0$. This implies that $\forall \tau_2\geq 0$, the equation \eqref{eq:stability_eq_general} has the same numbers of poles with negative real parts whenever $|\mu_k|=\sqrt{a_k^2+b_k^2}<1$. Therefore, the system \eqref{eq:consensus_with_delay} globally asymptotically achieves a consensus if $\max_{k=1,\ldots,d} |\mu_k| < 1$.
\end{proof}

Second, in case the cross-layer interactions are delay-free, i.e., $\tau_1>0$, $\tau_2 = 0$, we have the following theorem, which provides a sufficient consensus condition.
\begin{Theorem}\label{thm:cross-delay-free}
    Suppose that Assumptions~\ref{assumption:1} and \ref{assumption:2} hold, $|\mu_k|=|a_k+\jmath b_k| < 1$, $a_k \geq 0$,~$\forall k=1,\ldots,d$, and $\tau_2 = 0$. The system \eqref{eq:netw_matrix_form} globally asymptotically achieves a consensus if $0 \leq \tau_1< \frac{\pi}{2\lambda_{\max} (1+b_{\max})}$, where $\lambda_{\max} = \max_{i=1,\ldots,n}\lambda_i$ and $b_{\max} = \max_{k=1,\ldots,d}|b_k|$.
\end{Theorem}

\begin{proof}
Substituting $\tau_2 = 0$ and $\m{A}_{\rm cross}= \m{T}\m{J}\m{T}^{-1}$ into Eqn.~\eqref{eq:characteristic_equation}, and $\mu_k = a_k + \jmath b_k$, Eqn.~\eqref{eq:characteristic_equation} is equivalent to $d$ equations 
\begin{align} \label{eq:char_tau2_eq_0}
    s+\lambda_i(a_k + \jmath b_k) + \lambda_i e^{-\tau_1 s} = 0,~k=1,\ldots,d.
\end{align}
Substituting $s = \sigma + \jmath \omega$ into Eqn.~\eqref{eq:char_tau2_eq_0} yields
\begin{subequations} \label{eq:char_tau2_eq_0_a}
\begin{align}
    \sigma &= - \lambda_i a_k - \lambda_i e^{-\sigma\tau_1}\cos(\tau_1\omega), \label{eq:char_tau2_eq_0_a1}\\
    \omega &= -\lambda_i b_k - \lambda_i e^{-\sigma\tau_1}\sin(\tau_1\omega). \label{eq:char_tau2_eq_0_a2}
\end{align}
\end{subequations}
It follows from \eqref{eq:char_tau2_eq_0_a} that $(\sigma+\lambda_i a_k)^2 + (\omega+\lambda_i b_k)^2 = \lambda_i^2 e^{-2\sigma\tau_1}$. Thus, $ |\omega+\lambda_i b_k| \leq \lambda_i e^{-\sigma\tau_1}$, or
\begin{align}
    -\lambda_i (e^{-\sigma\tau_1} + b_k) \leq \omega \leq \lambda_i(e^{-\sigma\tau_1} - b_k).
\end{align}
Assume that $\sigma\ge 0$, then $e^{-\sigma\tau_1}\le 1$, we have
\begin{align}
    -\lambda_i (1 + b_k) &\leq \omega \leq \lambda_i(1 - b_k).
\end{align}
Since $|b_k| = \sqrt{|\mu_k|^2-a_k^2} < 1$, it follows that $\cos(\tau_1\omega) \geq \cos(\lambda_i(1+|b_k|)\tau_1) \geq \cos(\lambda_{\max} (1+b_{\max} )\tau_1) > \cos\left(\frac{\pi}{2}\right) = 0$. It follows from Eqn.~\eqref{eq:char_tau2_eq_0_a1}  that $\sigma < 0$, which contradicts our assumption that $\sigma \ge 0$. Thus, $\sigma < 0$, or equivalently, the system globally asymptotically achieves a consensus.
\end{proof}

Finally, we study the consensus on the multilayer network with two time delays. 
In the first step, we consider the case $\tau_1 = \tau_2 = \tau$ and determine the maximal time-delay $\tau_{\max}$ at which the system is marginally stable. Then, in the second step, we prove that the system globally asymptotically achieves a consensus for all $\tau_1, \tau_2 \geq 0$ that do not exceed a delay margin which is calibrated from $\tau_{\max}$.

\begin{Lemma} \label{lemma:same-time-delay}
    Suppose that Assumptions~\ref{assumption:1} and \ref{assumption:2} hold, $|\mu_k|<1$, $\mu_k=a_k+\jmath b_k$, $\forall k = 1,\ldots, d,$ and $\tau_1 = \tau_2 = \tau$. The system \eqref{eq:netw_matrix_form} globally asymptotically achieves a consensus if and only if $\tau < \tau_{\max} = \frac{c}{\lambda_{\min} \zeta_{\max}},$ where $\lambda_{\max} = \displaystyle\max_{i=1,\ldots,n}\lambda_i$, $\zeta_k=1+a_k+\jmath b_k$, $\zeta_{\max} = \max_{k=1,\ldots,d} |\zeta_k|$, $c = \min_{k=1,\ldots,d} \min\{|-\frac{\pi}{2}+\alpha_k|,|\frac{\pi}{2}+\alpha_k|\}$, and $\alpha_k = \arg\zeta_k$. 
\end{Lemma}
\begin{figure}
    \centering
    \includegraphics[width = 0.65\linewidth]{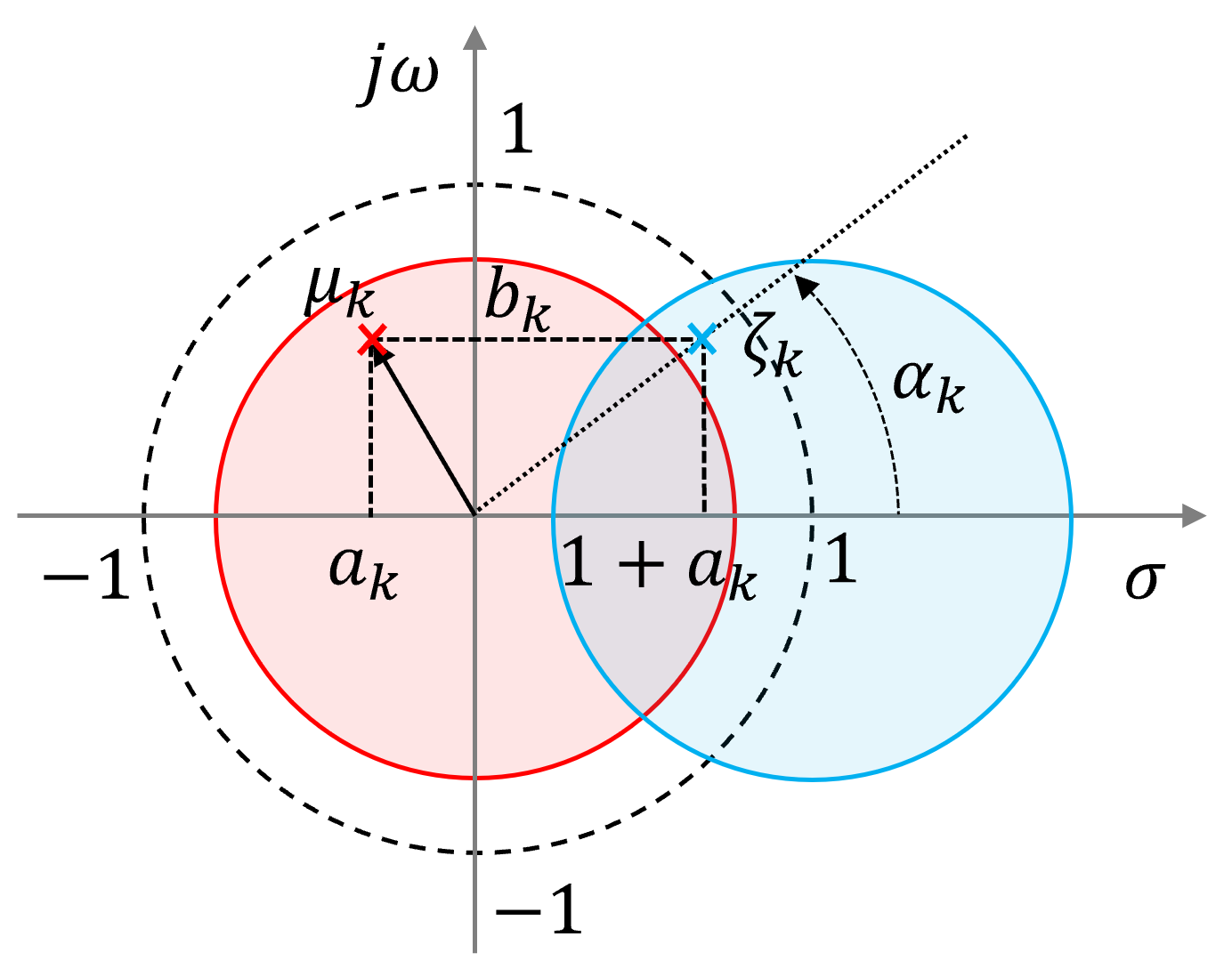}
    \caption{The eigenvalues of $\m{A}$ and $\m{A}_{\rm cross}$ are located in the blue circle and the red circle, respectively.}
    \label{fig:root_location}
\end{figure}
\begin{proof}
The equation $\text{det}((s + \lambda_i e^{-\tau_1 s})\m{I}_d + \lambda_i \m{A}_{\rm cross}e^{-\tau_2 s}) = 0$ is equivalent to $d$ equations 
\begin{subequations}
    \begin{align} \label{eq:stability_eq_general_1}
    s+ \lambda_i e^{-\tau_1 s}+ \lambda_i \mu_k e^{-\tau_2 s} &=0, \\
    s + \lambda_i (1+\mu_k) e^{-\tau s} &= 0, \label{eq:common_time_delay}
\end{align}
\end{subequations}
$\forall i=2,\ldots,n,~k=1,\ldots,n$, where we have substituted $\tau_1 = \tau_2 = \tau$ into Eqn.~\eqref{eq:common_time_delay}. Let $s = \sigma + \jmath \omega$ and $\mu_k = a_k + \jmath b_k$, the roots of Eqn.~\eqref{eq:common_time_delay} satisfy
\begin{subequations}
\begin{align}
\sigma &= -r_{ik} e^{-\tau\sigma} \cos(\tau\omega - \alpha_k), \\
\omega &= r_{ik} e^{-\tau\sigma} \sin(\tau\omega - \alpha_k),
\end{align}
\end{subequations}
where $r_{ik} = \lambda_i \sqrt{(1+a_k)^2 + b_k^2} = \lambda_i |\zeta_k|$, $\cos\alpha_k = \frac{1+a_k}{\sqrt{(1+a_k)^2 + b_k^2}}$, and $\sin\alpha_k = \frac{b_k}{\sqrt{(1+a_k)^2 + b_k^2}}$. As depicted in Fig.~\ref{fig:root_location}, we have $\alpha_k \in [0, \frac{\pi}{2})$. The system globally asymptotically achieves a consensus if and only if $\sigma<0,~\forall \omega$.

(Necessity) Suppose that $\sigma<0,~\forall \omega$, it follows that $\cos(\tau\omega-\alpha_k) = \cos(\tau r_{ik} e^{-\tau\sigma} \sin(\tau\omega - \alpha_k)-\alpha_k)>0,~\forall \omega$. This implies that
\[-\frac{\pi}{2}+\alpha_k<\tau r_{ik} e^{-\tau\sigma} \sin(\tau\omega - \alpha_k) < \frac{\pi}{2}+\alpha_k,~\forall \omega,\]
which is satisfied if
\begin{align*}
\tau r_{ik} e^{-\tau\sigma} <\min\left\{\left|-\frac{\pi}{2}+\alpha_k\right|, \left|\frac{\pi}{2}+\alpha_k\right|\right\}:=c_k.
\end{align*}
It follows that $\tau < \frac{c_k}{r_{ik}}e^{\tau\sigma}<\frac{c_k}{r_{ik}}.$ Since the condition should be held for all $i, k$, the consensus condition is given as
\begin{align} \label{eq:delay_consensus_condition}
\tau < \frac{c}{\lambda_{\max} \zeta_{\max}}.
\end{align}

(Sufficiency) Suppose that $\tau < \frac{c}{\lambda_{\max} \zeta_{\max}}$, because $\sigma^2 + \omega^2 = \lambda_i^2(1+\mu_k)^2 e^{-2\tau\sigma}$, it follows that $|\omega| \leq r_{ki} e^{-\tau\sigma}$. If $\sigma\geq 0$, then $e^{-\tau\sigma} \leq 1$ and $\tau|\omega| < c $. Thus, 
\[-c_k -\alpha_k\le -c -\alpha_k < \tau\omega - \alpha_k < c - \alpha_k \leq c_k - \alpha_k.\]
Note that 
\begin{align*}
    c_k = \begin{cases}
        \frac{\pi}{2} - \alpha_k, & \text{if } \alpha_k \geq 0,\\
        \frac{\pi}{2} + \alpha_k, & \text{if } \alpha_k < 0.
    \end{cases}
\end{align*}
$c_k - \alpha_k$ and $-c_k-\alpha_k$ both belong to $\left(-\frac{\pi}{2},\frac{\pi}{2} \right)$. It follows that $\cos(\tau\omega) > \min\{ \cos(c+\alpha_k),\cos(c+\alpha_k)\} \geq 0$, and $\sigma < 0$, which is a contradiction. Therefore,  $\sigma<0$.
\end{proof}

\begin{Theorem} \label{thm:general-system}
    Suppose that all assumptions of Lemma~\ref{lemma:same-time-delay} hold. Then, the system \eqref{eq:netw_matrix_form} globally asymptotically achieves a consensus if one of the following conditions is satisfied
    \begin{itemize}
        \item[(i)]  $0\leq \tau_1 \leq \tau_2 < \frac{\tau_{\max}}{\sqrt{2}}$, where $\tau_{\max}$ is given as in Lemma~\ref{lemma:same-time-delay}.
        \item[(ii)] $0 \leq \tau_1 \leq \tau_2 < \tau_{\max}'= \frac{c}{\lambda_{\max} \zeta_{\max}'},$ where   $\zeta_{\max}' = \max_{k=1,\ldots,d} (1+|a_k| +|b_k|)$, $\lambda_{\max}$ and $c$ are defined as in Lemma~\ref{lemma:same-time-delay}.
    \end{itemize}
\end{Theorem}

\begin{proof}
    We first consider the case $0\le \tau_1 \le \tau_2 < \tau_{\max}$. From the Eqn.~\eqref{eq:stability_eq_general_1}, it follows that
    \begin{subequations}
        \begin{align}
            \sigma &= - \lambda_i e^{-\sigma\tau_1}\cos(\omega\tau_1) \nonumber \\
            &\qquad\quad - \lambda_i e^{-\sigma\tau_2} (a_k \cos(\omega\tau_2) +b_k\sin(\omega\tau_2)),\label{eq:sigmaThm3}\\
            \omega &= \lambda_i e^{-\sigma\tau_1}\sin(\omega\tau_1) \nonumber \\
            &\qquad \quad + \lambda_i e^{-\sigma\tau_2} (a_k\sin(\omega\tau_2) - b_k \cos(\omega\tau_2)).
            \label{eq:omegaThm3}
        \end{align}
    \end{subequations}
    Suppose that $\sigma \geq 0$, we have $0<e^{-\sigma\tau_2} \leq e^{-\sigma\tau_1}\leq 1$. %Consider two cases:
    \begin{itemize}
    \item If the conditions in the statement (i) are satisfied, it follows from Eqn.~\eqref{eq:omegaThm3} that 
    \begin{align*}
        \omega^2 &\leq \lambda_i^2 \Big((e^{-\sigma\tau_1} \sin(\omega\tau_1) + a_ke^{-\sigma\tau_2}\sin(\omega\tau_2))^2 \\ &\qquad + e^{-2\sigma\tau_2}b_k^2 \cos^2(\omega\tau_2)\Big) (1^2 + 1^2 )\\
        &\leq 2\lambda_i^2 ((1+a_k)^2 + b_k^2)e^{-2\sigma\tau_1}
    \end{align*}
    It follows that $|\omega| \leq \sqrt{2}\lambda_i |\zeta_k|$, and thus,
    \begin{align*}
        0\leq \tau_1|\omega| \leq \tau_2|\omega| < \sqrt{2}\tau_{\max}\lambda_i |\zeta_k| \leq c < \frac{\pi}{2}.
    \end{align*}
    \item If the conditions in statement (ii) are satsified, it follows from Eqn.~\eqref{eq:omegaThm3} that
    $|\omega| \leq \lambda_i (1+|a_k| + |b_k|)$, and
    \begin{align*} %\label{eq:tau1tau2}
        0\leq \tau_1|\omega| \leq \tau_2|\omega| < \tau_{\max}'\lambda_i (1+|a_k| + |b_k|) \leq c < \frac{\pi}{2}.
    \end{align*}
    \end{itemize}
   As a result, in both (i) and (ii), we have
    \begin{align}
        \cos(\tau_1\omega) \geq \cos(\tau_2\omega) > \cos(c) > \cos\left(\frac{\pi}{2}\right) = 0.
    \end{align}
    Combining with \eqref{eq:sigmaThm3}, we have 
    \begin{align*}
        \sigma &< -\lambda_i e^{-\sigma\tau_2} ((1+a_k)\cos(\tau_2\omega) + b_k \cos(\tau_2\omega)) \\
        &= -r_{ik} e^{-\sigma\tau_2} \cos(\tau_2\omega - \alpha_k).
    \end{align*}
    Since $-c-\alpha_k < \tau_2 \omega - \alpha_i < c-\alpha_k$, it follows that $\cos(\tau_2\omega - \alpha_k)>0$. This implies $\sigma <0$, and we have a contradiction. 
    
    Therefore, if $0\leq \tau_1 \leq \tau_2 <\tau_{\max}$, we have $\sigma < 0$, i.e., the system  asymptotically achieves a consensus.
\end{proof}

\begin{Remark}
    In \cite{Ruan2003OnTZ}, a two-time-delay system 
    \begin{align} \label{eq:existing_result}
        \dot{x}(t) = -ax(t) - b(x(t-\tau_1) + x(t-\tau_2)),
    \end{align}
    where $a>0, 0<\tau_1<\tau_2$, has been studied. The authors gave a necessary and sufficient condition for stability of \eqref{eq:existing_result} by specifying a bound of the parameter $b$, which depends on the solution of a transcendental equation. In this paper, the characteristic equation~\eqref{eq:stability_eq_general} has a similar form to \eqref{eq:existing_result}. However, the coefficients associated with the delay terms are not identical and the approach in \cite{Ruan2003OnTZ} is inapplicable. Although Theorem~\ref{thm:general-system} provides only a sufficient condition for asymptotic convergence, our analysis is simpler and relies on only simple computations.
\end{Remark}
\begin{figure*}[t]
    \centering
\subfloat[$\tau_1 = 0, \tau_2=2$]{\includegraphics[width=0.32\linewidth]{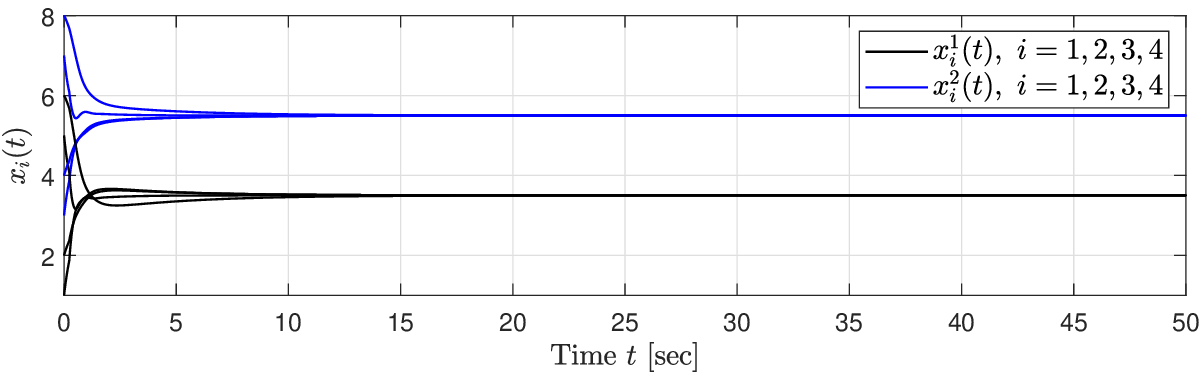}}\hfill
\subfloat[$\tau_1 = 0, \tau_2=5$]{\includegraphics[width=0.33\linewidth]{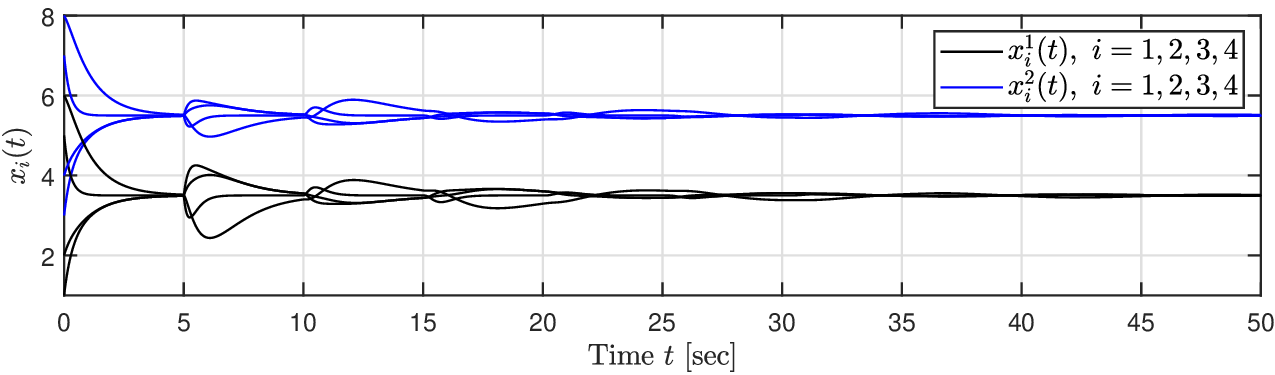}}\hfill
\subfloat[$\tau_1 = 0, \tau_2=10$]{\includegraphics[width=0.32\linewidth]{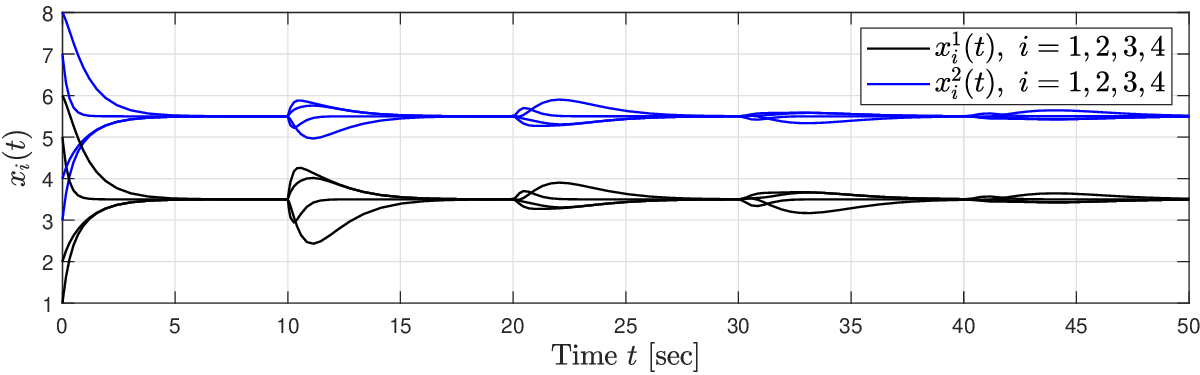}}
\caption{Simulations of delayed two-layer consensus network of four agents \eqref{eq:consensus_with_delay} with $\tau_1=0$ and the matrix $\m{A}$ has $|\mu_k|<1$.}
\label{simul:twolayer_intralayer_delayfree}
\end{figure*}
\section{Two-layer networks and simulation results}
\label{sec:twolayer_network}
In this section, we focus on two-layer matrix-weighted networks with time delays and present the corresponding simulation results. The associated matrix that captures the agent-to-agent interaction pattern is given by
$\m{A} = \begin{bmatrix}
    1 & a^{12}\\
    a^{21} & 1
\end{bmatrix}$. Thus, 
$\m{A}_{\rm cross}=-\m{I}_2+\m{A}=\begin{bmatrix}
     0 & a^{12}\\
     a^{21} & 0 \end{bmatrix}$ has a pair of pure imaginary (real) eigenvalues when $a^{12}a^{21}<0$ (resp., $a^{12}a^{21}>0$), and the following result can be stated. 
\begin{figure*}[t]
\centering
\subfloat[$\tau_1 = 0, \tau_2=1$]{\includegraphics[width=0.32\linewidth]{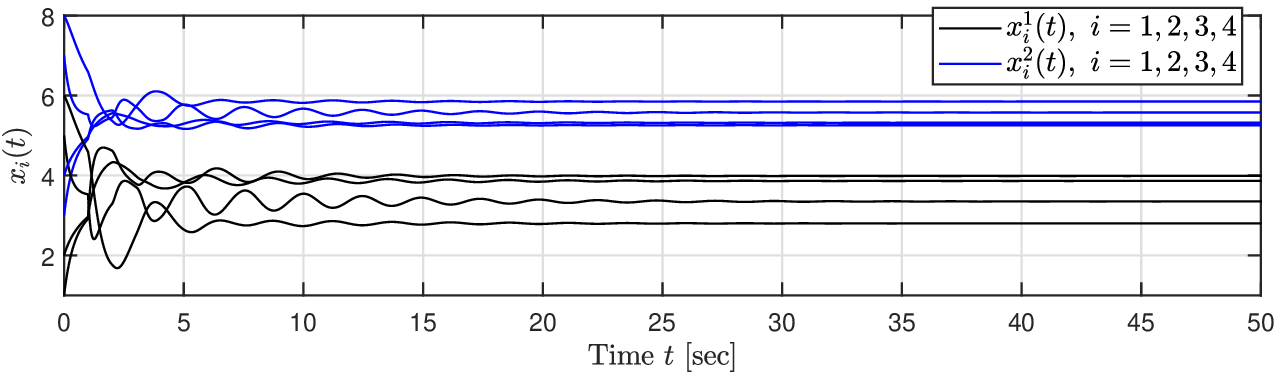}}\hfill
\subfloat[$\tau_1 = 0, \tau_2=2$]{\includegraphics[width=0.32\linewidth]{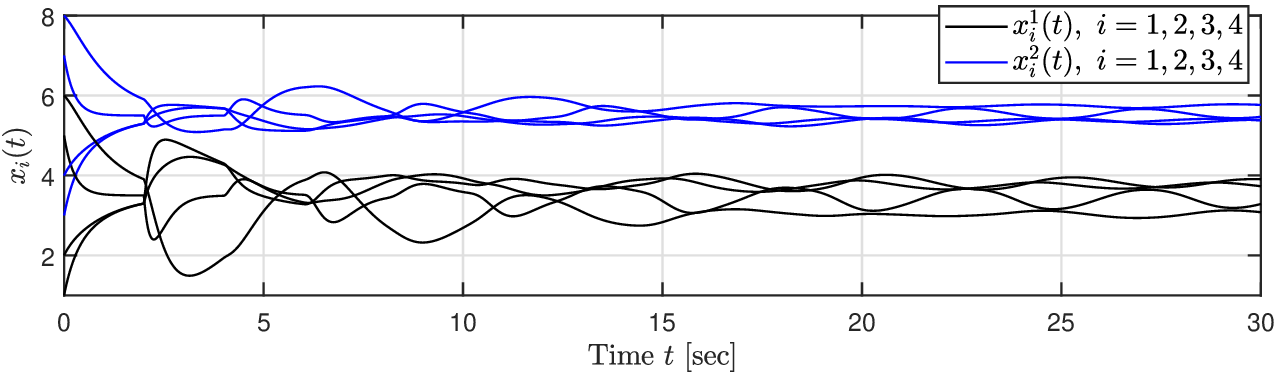}}\hfill
\subfloat[$\tau_1 = 0, \tau_2=10$]{\includegraphics[width=0.32\linewidth]{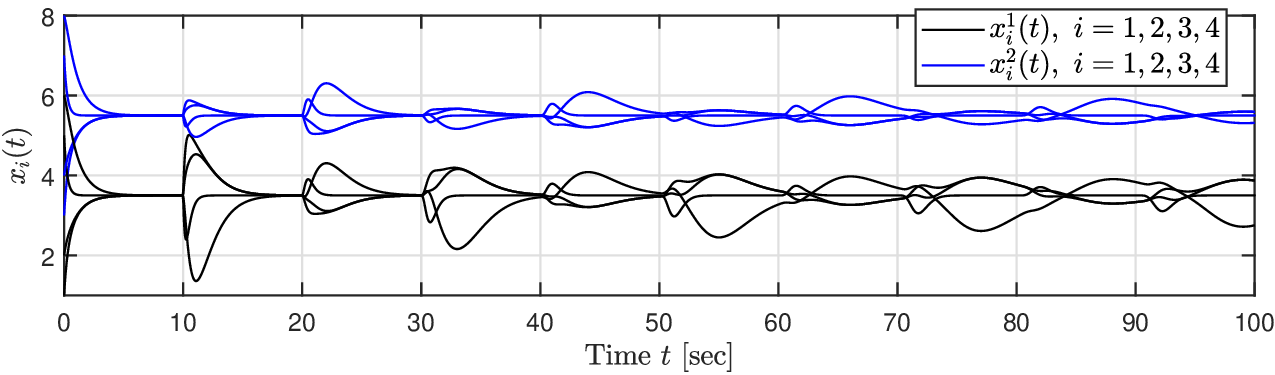}}
\caption{Simulations of delayed two-layer consensus network of four agents \eqref{eq:consensus_with_delay} with $\tau_1=0$ and the matrix $\m{A}$ has an eigenvalue $|\mu_k| = 1$.}
\label{simul:twolayer_intralayer_delayfree_A2}
\end{figure*}

\begin{figure*}[t]
    \centering
\subfloat[$\tau_1 = 0, \tau_2=2$]{\includegraphics[width=0.32\linewidth]{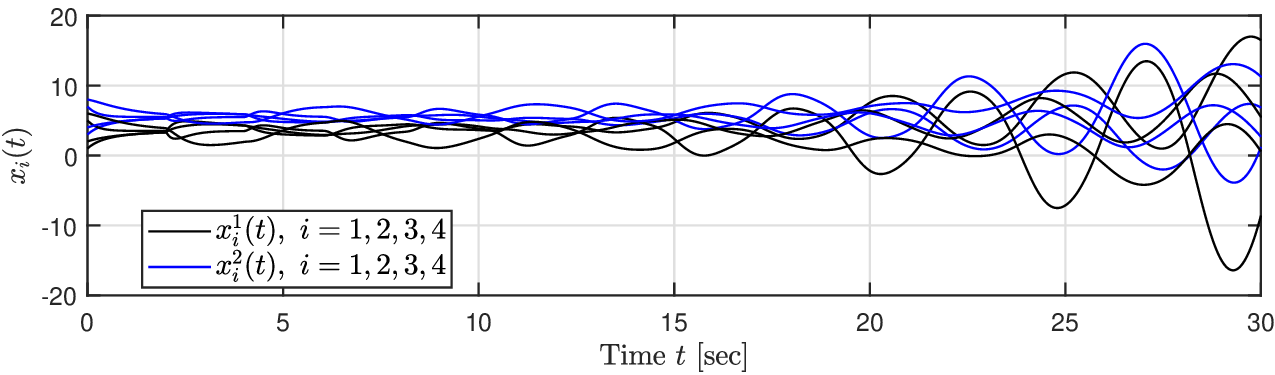}}\hfill
\subfloat[$\tau_1 = 0, \tau_2=5$]{\includegraphics[width=0.32\linewidth]{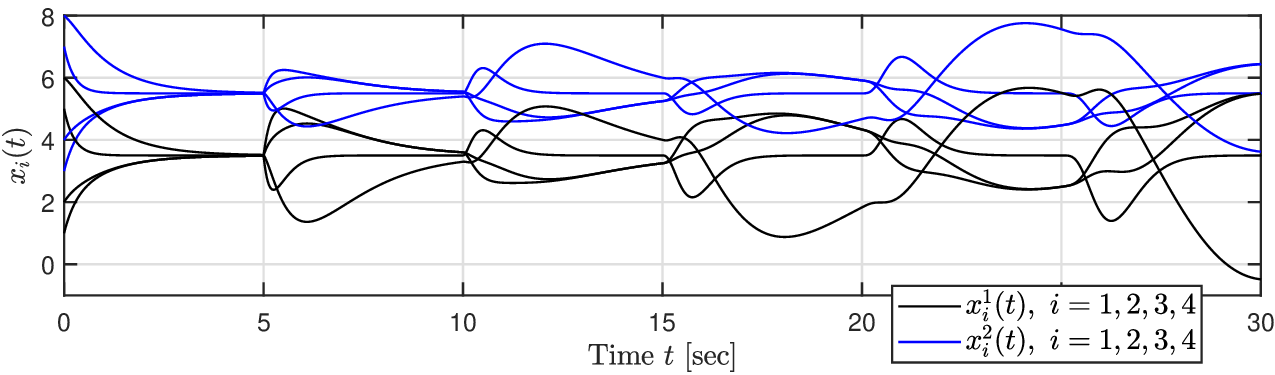}}\hfill
\subfloat[$\tau_1 = 0, \tau_2=10$]{\includegraphics[width=0.32\linewidth]{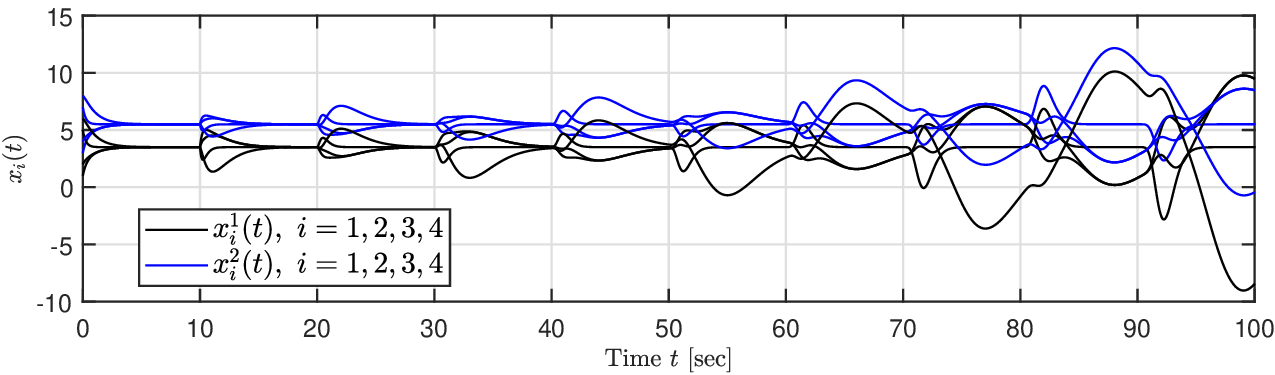}}
\caption{Simulations of delayed two-layer consensus network of four agents \eqref{eq:consensus_with_delay} with $\tau_1=0$ and the matrix $\m{A}$ has eigenvalues $|\mu_k| > 1$.}
\label{simul:twolayer_intralayer_delayfree_A3}
\end{figure*}

\begin{figure*}[h!]
\centering
\subfloat[$\tau_1 = 0.23, \tau_2=0.23$]{\includegraphics[width=.32\linewidth]{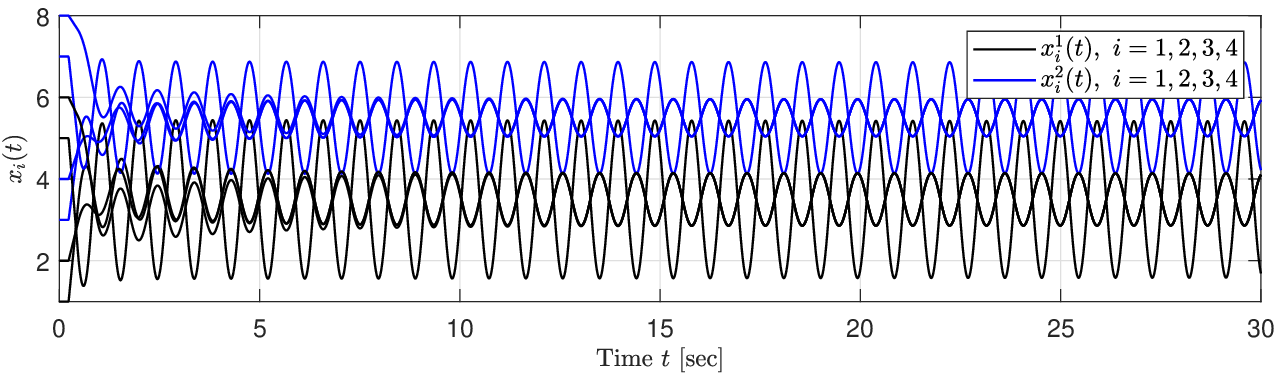}}\hfill
\subfloat[$\tau_1 = 0.2, \tau_2=0.2$]{\includegraphics[width=.32\linewidth]{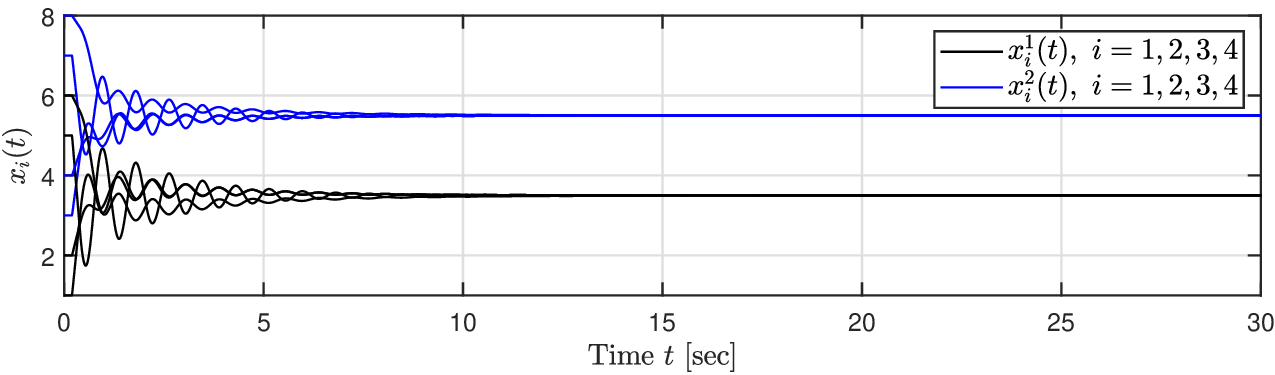}} \hfill
\subfloat[$\tau_1 = 0.2, \tau_2=0.23$]{\includegraphics[width=.31\linewidth]{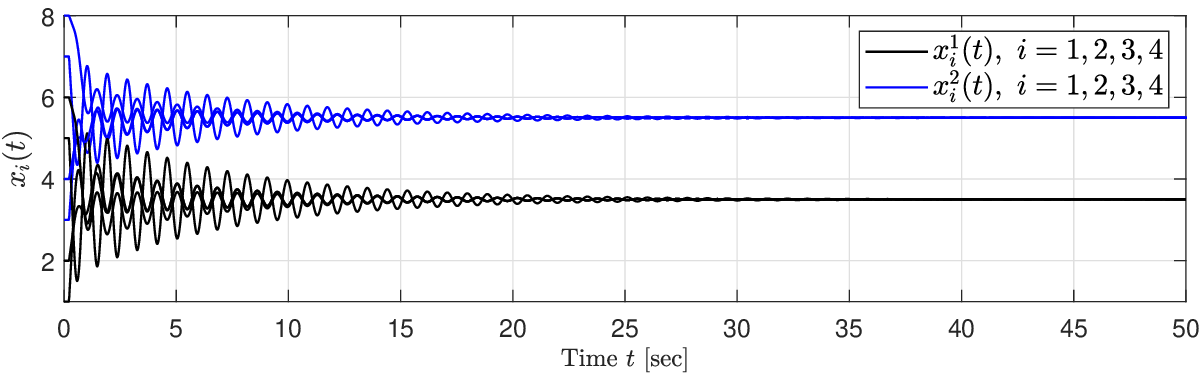}}\\
\subfloat[$\tau_1 = 0.23, \tau_2=2$]{\includegraphics[width=.32\linewidth]{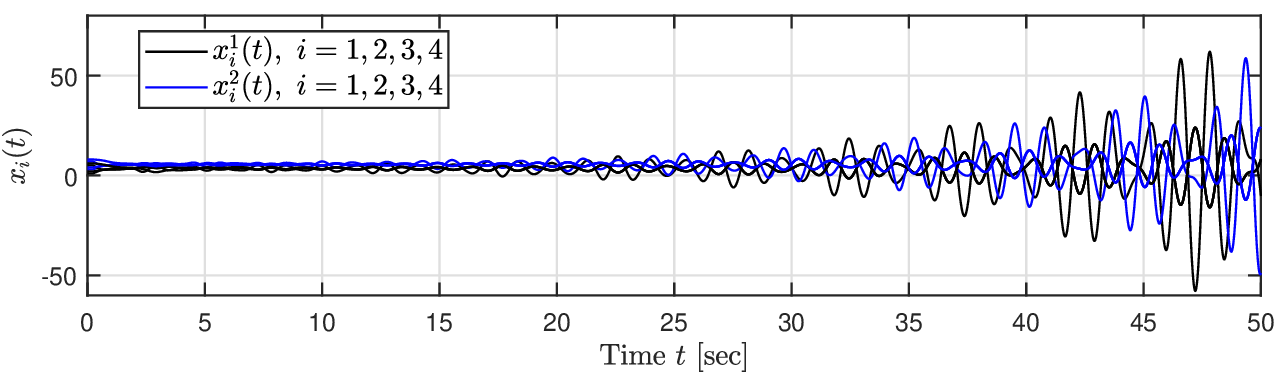}}\hfill
\subfloat[$\tau_1 = 0.23, \tau_2=10$]{\includegraphics[width=.31\linewidth]{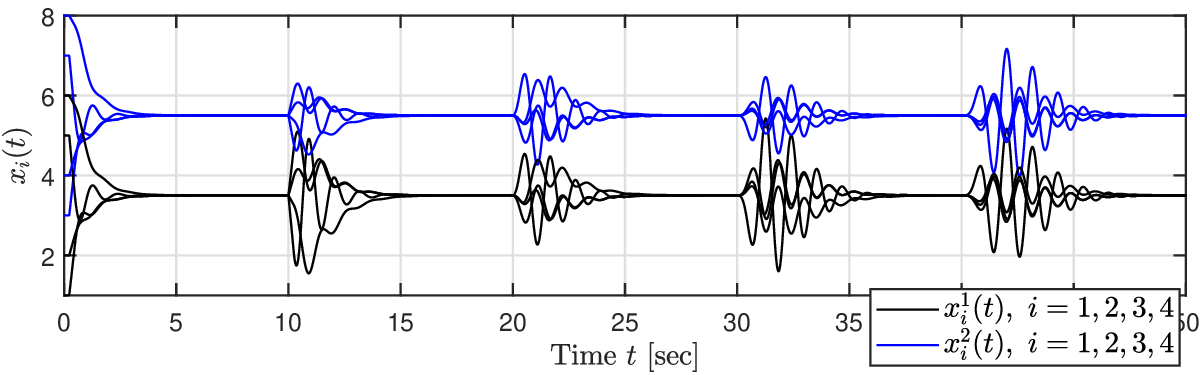}}\hfill
\subfloat[$\tau_1 = 0.5, \tau_2=0.23$]{\includegraphics[width=.32\linewidth]{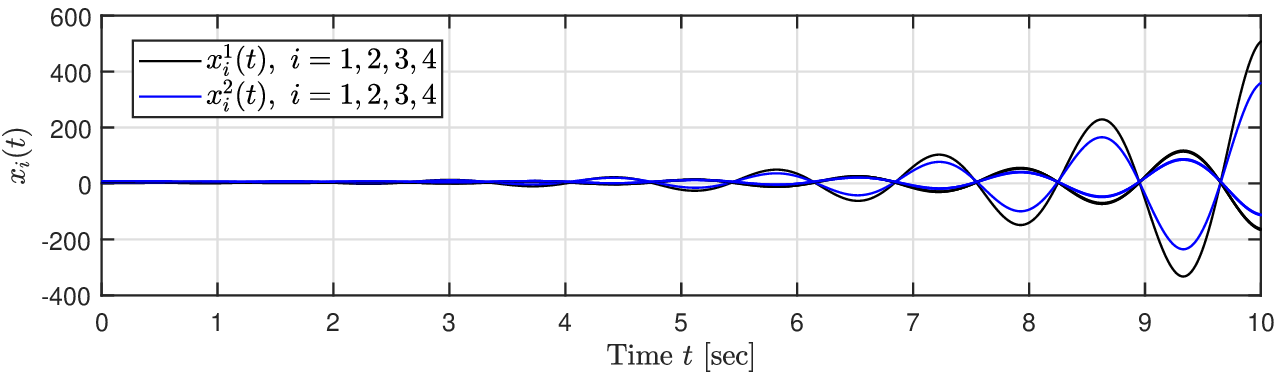}}
\caption{Simulations of delayed 4-agent 2-layer network \eqref{eq:consensus_with_delay} with $\tau_1,\tau_2 \ge 0$.}
\label{simul:twolayer_twodelays}
\end{figure*}
\begin{Corollary}\label{cor:2layer_1}
Suppose that Assumptions~\ref{assumption:1} and \ref{assumption:2} hold, and $|a^{12}a^{21}| < 1$. Then, the two-layer consensus system globally asymptotically achieves a consensus if
\begin{itemize}
\item[(i)] $\tau_1=0$ (and this is also a necessary condition); or
\item[(ii)] $\tau_2=0$, and $0\leq \tau_1< \frac{\pi}{2\lambda_{\max}(1+\sqrt{|a^{12}a^{21}}|)} := \tau_{\max}$; or
\item [(iii)] $-1 < a^{12}a^{21}<0$, $0\leq \tau_1 \leq \tau_2 < \tau_{\max}$. 
\end{itemize}
\end{Corollary}

Next, consider a network of 4 agents having the interaction graph as depicted in Fig.~\ref{fig:2layerNetw}. 
% We have 
% {\small 
% \begin{align*}
%     \m{L} = \begin{bmatrix}
%     2& -1& -1& 0\\
%     -1&3&-1&-1\\
%     -1&-1&2&0\\
%     0&-1&0&1  
% \end{bmatrix}.
% \end{align*}}

{\flushleft \emph{Simulations for intra-layer delay-free two-layer networks ($\tau_1 = 0$): }} We conduct several simulations of the consensus algorithm with $\m{A}_1 = \begin{bmatrix}
    1 & 1\\
    0.5 & 1
\end{bmatrix}$. In this case, $\m{A}$ has eigenvalues satisfying $|\mu_k|<1,~k=1,2$. The simulation results for $\tau_2 = 2, 5, 10$ are shown in Fig.~\ref{simul:twolayer_intralayer_delayfree}. Clearly, $x_i^k,~k=1,2,$ consent on two values (consensus states) in all three cases. 

Next, consider $\m{A}_2 = \begin{bmatrix}
    1 & 2\\
    0.5 & 1
\end{bmatrix}$. Correspondingly, $|\mu_k| = 1$. Simulation results in Fig.~\ref{simul:twolayer_intralayer_delayfree_A2} show that cross-layer time delays do not destabilize the system but perturb the system from the consensus set.

Finally, we consider $\m{A}_2 = \begin{bmatrix}
    1 & 2\\
    1 & 1
\end{bmatrix}$. In this case, $|\mu_k|>1$. Simulation results in Fig.~\ref{simul:twolayer_intralayer_delayfree_A3} show that cross-layer time delays destabilize the consensus system. 

{\flushleft \emph{Simulations for two time-delay networks ($\tau_1, \tau_2 \geq 0$): }} We simulate the two-layer network under the presence of two time-delays $\tau_1, \tau_2 \ge 0$ with matrix $\m{A}_1$. Corresponding to Corollary~1, we have $\tau_{\max} = 0.23$. The simulation result in Fig.~\ref{simul:twolayer_twodelays} shows that the states $x_i^k,~k = 1, 2,$ approach to two sinusoidal trajectories for $\tau_1 = \tau_2 = \tau_{\max}$ (i.e., the system has a pair of purely imaginary eigenvalues), two common constant values (or the consensus state) for $0< \tau_1 \le \tau_2 \leq \tau_{\max}$, and is unstable in case $\tau_1 = \tau_{\max}< \tau_2$ (Fig.~\ref{simul:twolayer_twodelays}(d,e)), or  $\tau_2 = \tau_{\max} <\tau_1$ (Fig.~\ref{simul:twolayer_twodelays}(f)). 

Thus, the simulation results are consistent with the theoretical analysis.

\section{Conclusions}
\label{sec:conclusion}
In this paper, we derived several consensus conditions for multilayer networks characterized by repeated agent-to-agent interaction patterns and two distinct time delays in both intra-layer and cross-layer interactions. Additionally, specific conditions were provided for the case of two-layer networks. Future work will address time delays in multilayer networks with more general interaction patterns and diverse network topologies.

\bibliographystyle{IEEEtran}
\bibliography{minh2024}  

\begin{thebibliography}{10}
\providecommand{\url}[1]{#1}
\csname url@samestyle\endcsname
\providecommand{\newblock}{\relax}
\providecommand{\bibinfo}[2]{#2}
\providecommand{\BIBentrySTDinterwordspacing}{\spaceskip=0pt\relax}
\providecommand{\BIBentryALTinterwordstretchfactor}{4}
\providecommand{\BIBentryALTinterwordspacing}{\spaceskip=\fontdimen2\font plus
\BIBentryALTinterwordstretchfactor\fontdimen3\font minus
  \fontdimen4\font\relax}
\providecommand{\BIBforeignlanguage}[2]{{%
\expandafter\ifx\csname l@#1\endcsname\relax
\typeout{** WARNING: IEEEtran.bst: No hyphenation pattern has been}%
\typeout{** loaded for the language `#1'. Using the pattern for}%
\typeout{** the default language instead.}%
\else
\language=\csname l@#1\endcsname
\fi
#2}}
\providecommand{\BIBdecl}{\relax}
\BIBdecl

\bibitem{Olfati2007consensuspieee}
Olfati-Saber, Fax, and Murray, ``Consensus and cooperation in networked
  multi-agent systems,'' \emph{Proceedings of the IEEE}, vol.~95, no.~1, pp.
  215--233, 2007.

\bibitem{He2017TSMC}
He \emph{et~al.}, ``Multiagent systems on multilayer networks:
  {S}ynchronization analysis and network design,'' \emph{IEEE Transactions on
  Systems, Man, and Cybernetics: Systems}, vol.~47, no.~7, pp. 1655--1667,
  2017.

\bibitem{Lee2017consensus}
Lee and Shim, ``Consensus of linear time-invariant multi-agent system over
  multilayer network,'' in \emph{Proc. of the 17th Internat Conf Control
  Automat and Syst (ICCAS)}.\hskip 1em plus 0.5em minus 0.4em\relax IEEE, 2017,
  pp. 133--138.

\bibitem{Sorrentino2020group}
Sorrentino, Pecora, and Trajkovi{\'c}, ``Group consensus in multilayer
  networks,'' \emph{IEEE Transactions on Network Science and Engineering},
  vol.~7, no.~3, pp. 2016--2026, 2020.

\bibitem{Trinh2018Aut}
Trinh \emph{et~al.}, ``Matrix-weighted consensus and its applications,''
  \emph{Automatica}, vol.~89, pp. 415--419, 2018.

\bibitem{Ahn2020opinion}
Ahn \emph{et~al.}, ``Opinion dynamics with cross-coupling topics: Modeling and
  analysis,'' \emph{IEEE Transactions on Computational Social Systems}, vol.~7,
  no.~3, pp. 632--647, 2020.

\bibitem{Ye2020Aut}
Ye \emph{et~al.}, ``Continuous-time opinion dynamics on multiple interdependent
  topics,'' \emph{Automatica}, vol. 115, no. 108884, 2020.

\bibitem{Trinh2024networked}
Trinh, Le-Phan, and Ahn, ``The networked input-output economic problem,''
  \emph{arXiv preprint arXiv:2412.13564}, 2024.

\bibitem{Fridman2014tutorial}
E.~Fridman, ``Tutorial on {L}yapunov-based methods for time-delay systems,''
  \emph{European Journal of Control}, vol.~20, no.~6, pp. 271--283, 2014.

\bibitem{Sun2008average}
Sun, Wang, and Xie, ``Average consensus in networks of dynamic agents with
  switching topologies and multiple time-varying delays,'' \emph{Systems \&
  Control Letters}, vol.~57, no.~2, pp. 175--183, 2008.

\bibitem{Cepeda2011exact}
Cepeda-Gomez and Olgac, ``An exact method for the stability analysis of linear
  consensus protocols with time delay,'' \emph{IEEE Transactions on Automatic
  Control}, vol.~56, no.~7, pp. 1734--1740, 2011.

\bibitem{Singh2015synchronization}
Singh \emph{et~al.}, ``Synchronization in delayed multiplex networks,''
  \emph{Europhysics Letters}, vol. 111, no.~3, p. 30010, 2015.

\bibitem{Lam2024consensus}
Pham \emph{et~al.}, ``Consensus over matrix-weighted networks with
  time-delays,'' 2024.

\bibitem{Ruan2003OnTZ}
Ruan and Wei, ``On the zeros of transcendental functions with applications to
  stability of delay differential equations with two delays,'' \emph{Dynamics
  of Continuous, Discrete and Impulsive Systems Series {A}: {M}athematical
  Analysis}, vol.~10, pp. 863--874, 2003.

\end{thebibliography}

\end{document}